\documentclass[12pt,a4paper,twoside]{amsart}
\pdfoutput=1
\usepackage{amssymb}
\usepackage{mathtools}
\usepackage{hyperref}

\mathtoolsset{showonlyrefs}

\usepackage[english]{babel}
\usepackage[latin1]{inputenc}

\usepackage{mathrsfs}

\theoremstyle{plain}
\newtheorem{theorem}{Theorem}[section]
\newtheorem{lemma}[theorem]{Lemma}

\newtheorem{proposition}[theorem]{Proposition}

\newtheorem{problem}[theorem]{Problem}
\theoremstyle{definition}
\newtheorem{definition}[theorem]{Definition}

\theoremstyle{remark}
\newtheorem{remark}[theorem]{Remark}

\newcommand{\R}{\mathbb{R}}
\newcommand{\Z}{\mathbb{Z}}
\newcommand{\C}{\mathbb{C}}

\newcommand{\der}{\mathrm{d}}
\newcommand{\Der}[1]{\frac{\der}{\der #1}}
\newcommand{\eps}{\varepsilon}
\renewcommand{\phi}{\varphi}
\newcommand{\abs}[1]{\left| #1 \right|}
\newcommand{\aabs}[1]{\left\| #1 \right\|}
\DeclareMathOperator{\spt}{spt}
\DeclareMathOperator{\tr}{tr}

\newcommand{\SH}{\mathcal{SH}}
\newcommand{\U}{\mathcal{U}}
\newcommand{\ad}{\operatorname{ad}}
\newcommand{\Ad}{\operatorname{Ad}}
\newcommand{\im}{\operatorname{im}}
\newcommand{\initial}{\mathscr{I}}
\newcommand{\final}{\mathscr{F}}
\newcommand{\Texp}{\mathcal{T}\!\exp}
\newcommand{\rt}{\mathcal{I}}
\renewcommand{\Re}{\operatorname{Re}}

\newcommand{\order}{o}
\DeclareMathOperator{\Lip}{Lip}
\DeclareMathOperator{\id}{id}

\title{Coherent quantum tomography}
\author{Joonas Ilmavirta}
\address{Department of Mathematics and Statistics, University of Jyv\"askyl\"a, P.O.Box 35 (MaD) FI-40014 University of Jyv\"askyl\"a, Finland}
\email{joonas.ilmavirta@jyu.fi}

\begin{document}

\begin{abstract}
We discuss a quantum mechanical indirect measurement method to recover a position dependent Hamilton matrix from time evolution of coherent quantum mechanical states through an object. A mathematical formulation of this inverse problem leads to weighted X-ray transforms where the weight is a matrix. We show that such X-ray transforms are injective with very rough weights. Consequently, we can solve our quantum mechanical inverse problem in several settings, but many physically relevant problems we pose also remain open. We discuss the physical background of the proposed imaging method in detail. We give a rigorous mathematical treatment of a neutrino tomography method that has been previously described in the physical literature.
\end{abstract}




\maketitle

\section{Introduction}
\label{sec:intro}

We present a new imaging modality which we call coherent quantum tomography.
The reason for this name will be discussed in section~\ref{sec:phys}.
Our goal is to model and solve the following problem:
An object $\Omega\subset\R^n$ is to be imaged nondestructively.
Properties of the object are described by a (hermitean) matrix valued function~$H$ on~$\Omega$, a quantum mechanical Hamiltonian.
Particles are fired through~$\Omega$ and they experience time evolution due to the Hamiltonian~$H$ and the final state of the particle is measured.
How much about the function~$H$ can we infer from this data?

Of course, much depends on what initial states in the state space are possible and what kinds of measurements are made.
We will mainly focus on two kinds of situations, a general case with ideal data and a more realistic one with neutrino oscillations, but the discussion of section~\ref{sec:qm-intro} is valid for any data.
These problems will be described in more detail below.

The imaging modality we discuss is a quantum mechanical one and relies on coherence between states in a multiple dimensional state space.
We restrict our attention to finite dimensional state spaces, but many practically interesting systems are indeed finite dimensional.
The practical setting resembles that of X-ray tomography (sending rays through an object and measuring what happens to them) but the physical phenomena behind the measurement and the mathematical model are rather different.
Our analysis does, however, lead to a certain type of X-ray transform.

The application we have in mind is imaging the Earth with neutrinos.
The idea is to produce neutrinos somewhere at the surface (in a nuclear reactor or a particle accelerator) and measure neutrinos of different flavours coming from this source.
Neutrinos come in three flavours, electron, muon and tau neutrinos ($\nu_e,\nu_\mu,\nu_\tau$) and typical nuclear reactions produce electron neutrinos or electron antineutrinos.
These neutrino flavours will oscillate --- turn into each other while propagating --- and details of this oscillation depend on the medium they traverse.
The problem is then to recover properties of the medium from neutrino oscillation data.
For basics of neutrino physics, we refer to the book~\cite{GK:nu-book} and section~\ref{sec:nu}.

Neutrino tomography of the Earth is not a new idea.
Tomography of our planet using neutrino absorption has been discussed in~\cite{RS:nu-abs-tomography,JRF:nu-abs-tomography,RJF:nu-abs-tomography} and using matter effects in neutrino oscillations in~\cite{NJT:nu-tomography-earth,ML:nu-tomography,LOTW:sn-nu-tomography}.
A discussion and comparison of these two methods can be found in~\cite{W:nu-tomography}.
It has also been proposed that geoneutrinos produced in nuclear decays in the Earth could be used for imaging the Earth~\cite{MSBFWM:earth-antinu-tomography} and that neutrino tomography could be used as a means of earthquake prediction~\cite{WCL:nu-earthquake}.
A spectrometric approach to using neutrinos in imaging the Earth was recently proposed in~\cite{RTB:nu-earth}.
We aim to bridge the gap between the physical literature of the idea of neutrino tomography and the mathematical literature of inverse problems and imaging by providing a physically meaningful and mathematically rigorous discussion of the neutrino oscillation tomography problem.

Due to the quantum mechanical nature of the problem, our data is phaseless.
For some other results concerning phaseless data, we refer to~\cite{N:phaseless,W:helmholtz,IW:acoustic-hamilton}.

Problems considered in this paper can also be posed on manifolds, replacing~$\bar\Omega$ by a compact Riemannian manifold with boundary and lines by geodesics.
Euclidean geometry is most relevant for many physical applications, so for the sake of simplicity we restrict our attention to it.
A reader familiar with differential geometry will observe that many arguments carry over to Riemannian manifolds.
We end up introducing some new ray transform problems (see problems~\ref{prob:matrix-wxrt} and~\ref{prob:tensor}), which can also be asked on Riemannian manifolds.
We provide a partial solution in the Euclidean setting in theorems~\ref{thm:wrt} and~\ref{thm:wrt-local}.
See section~\ref{sec:wxrt} for more details on ray transforms.

\subsection{Outline}

We will introduce the quantum mechanical problem in detail in section~\ref{sec:qm-intro}, and then we will have the language to state our main results in section~\ref{sec:results}.
We will overview the physical aspects of the proposed indirect measurement in section~\ref{sec:phys}.
Section~\ref{sec:gauge} is devoted to converting our phase free data to a more useful form.
In section~\ref{sec:pseudolin} we will prove some of our main results using a pseudolinearization argument.
This will involve some weighted X-ray transforms where the weight is a matrix; these are introduced in section~\ref{sec:wxrt-intro}.
In section~\ref{sec:ideal-data} we will reduce our quantum mechanical problem to a ray transform problem.
We will also discuss the linearized problem in section~\ref{sec:lin}.
Our quantum mechanical problem involves a time ordered exponential, and we investigate the same problem without time ordering in section~\ref{sec:unordered}.
In section~\ref{sec:wxrt} we will show that some of our new weighted X-ray transforms are injective with very rough weights.

\section{The quantum mechanical problem}
\label{sec:qm-intro}

We will use natural units: $c=\hbar=1$.
Consider a point particle moving in $n$\mbox{-}dimensional Euclidean space, with position at a time~$t$ given by $\gamma(t)=x_0+tv\in\R^n$, where $\abs{v}=1$.
The particle is ultrarelativistic --- its speed is practically the speed of light --- and it is assumed not to scatter or be absorbed.

Besides moving in space, the particle has a state in a quantum mechanical state space~$\C^N$ where time evolution is given by the Schr\"odinger equation.
Vectors in the state space~$\C^N$ will be referred to as states.

Note that the dimension~$n$ of the ambient space and the dimension of the state space~$N$ need not have anything to do with each other.
The dimension~$N$ need not be high to accommodate interesting physical systems.
The case $N=1$ is in a sense empty, but two state systems (such as spin) provide an example where $N=2$ and for neutrino oscillations $N=3$ is relevant (excluding possible sterile neutrinos).

Suppose that the Hamiltonian governing the time evolution in the quantum mechanical state space~$\C^N$ is a hermitean $N\times N$ matrix~$H(x)$ depending on the position of the particle in~$\R^n$.
In this case the Schr\"odinger equation for the state $\Psi(t)\in\C^N$ is
\begin{equation}
i\partial_t\Psi(t)
=
H(\gamma(t))\Psi(t).
\end{equation}
For notational simplicity, let us drop the line~$\gamma$ and consider~$H$ a function of time.
Let $U_H(t_2,t_1)$ be the time evolution operator (solution operator) for which $\Psi(t)=U_H(t,t_0)\Psi_0$ is the unique solution to
\begin{equation}
\begin{cases}
i\Psi'(t)
=
H(t)\Psi(t)
\\
\Psi(t_0)=\Psi_0.
\end{cases}
\end{equation}
Formally, we may write~$U_H$ as a time ordered exponential\footnote{For a discussion of the time evolution operator as a time ordered exponential, see textbooks in quantum mechanics, such as~\cite{SN:qm-book}. The time ordered exponential has a series representation known as the Dyson series.}
\begin{equation}
\label{eq:Texp}
U_H(t_2,t_1)
=
\Texp\left(-i\int_{t_1}^{t_2}H(\gamma(t))\der t\right).
\end{equation}
If $[H(t),H(s)]=0$ for all $t,s\in[t_1,t_2]$, then this is just the matrix exponential.
We will consider unordered time evolution, where~$\Texp$ is replaced by~$\exp$, in section~\ref{sec:unordered}.

Consider the time evolution from creation of a particle at $t=0$ to observation at $t=T$.
Suppose there is a set $\initial\subset\C^N$ of possible initial states~$\Psi(0)$ and a set $\final\subset\C^N$ of observable final states.
We assume that we can prepare the particle to any initial state in~$\initial$ at will.
If the final state $\Psi(T)=U_H(T,0)\Psi(0)$ is measured via an operator~$A$, we can observe (from the statistics of a large number of measurements) the norms of the projections of~$\Psi(T)$ to eigenspaces of~$A$.
If the spectrum of~$A$ is nondegenerate, we can measure~$\abs{a^*\Psi(T)}^2$ for all eigenvectors (eigenstates)~$a$ of~$A$.
The set~$\final$ consists of all nondegenerate eigenvectors of all observables that can be used for measurements.

Consider a domain $\Omega\subset\R^n$ which we cannot enter but through which we can fire particles.
Our aim is to obtain information about the Hamiltonian~$H$ in~$\Omega$ from measurements of particles sent through it.
This is a nondestructive measurement problem or an inverse problem.

We assumed the particle to be ultrarelativistic ($\abs{v}=1$) and also scattering and absorption to be negligible.
This means essentially that the medium only affects the particle in its time evolution in the state space~$\C^N$.
As mentioned in section~\ref{sec:intro}, the main example we have in mind is neutrinos travelling through the Earth, the objective being to recover the inner structure of the planet from neutrino oscillation measurements.
Information from such measurements should, of course, be combined with similar measurements by other methods (most importantly seismic imaging), but here we focus on the quantum mechanical tomography problem.

We adopt the following notation:
For a line segment $\gamma\colon[0,T]\to\bar\Omega$ we denote $U^\gamma_H\coloneqq U_H^\gamma(T,0)$, the time evolution operator over the entire length of the line.

The problem is now this:
\begin{problem}
\label{prob:main}
For every line segment $\gamma\colon[0,T]\to\bar\Omega$ with endpoints on~$\partial\Omega$ we measure~$\abs{\Phi^*U_H^\gamma\Psi}$ for all $\Psi\in\initial$ and $\Phi\in\final$.
How much can we infer about the Hamiltonian $H\colon\Omega\to\C^{N\times N}$ from this data?
\end{problem}

We will assume throughout that the Hamiltonian only depends on the position.
It is also possible (and physically relevant, see section~\ref{sec:tensor}) that the Hamiltonian depends on the direction~$v$.
If the dependence on~$v$ is polynomial, the X-ray tomography problems that arise below will have to be generalized to tensor tomography problems.
Recovering a tensor field from its integrals over all lines (or geodesics) can only be possible up to potential tensors.
This natural gauge condition is the only obstruction in many situations; see for example~\cite{S:tensor-book,S:rce-tensor,PSU:tensor-survey,PSU:anosov-mfld}.

\begin{definition}
\label{def:ideal}
We say that the pair of subsets $\initial,\final\subset\C^N$ gives \emph{ideal data} if for any two unitary matrices~$U$ and~$V$ the condition $\abs{a^*Ub}=\abs{a^*Vb}$ for all $a\in\final$ and $b\in\initial$ implies $U=e^{i\phi}V$ for some $\phi\in\R$.
\end{definition}

We also assume perfect data in the sense that the measurements are done over every line through the domain.
Partial data problems with restricted sets of lines are also physically relevant, but we do not pursue them here.

The first problem is inferring as much as possible about~$U_H^\gamma$ from the data.
It follows from lemma~\ref{lma:unitary-phase} below that in the case of ideal data we can recover the matrix~$U_H^\gamma$ up to a phase factor.
This implies that if we add a position dependent scalar multiple of the identity matrix,~$f(x)I$, to~$H(x)$, the data is not changed.
We are therefore completely blind to multiples of the identity in~$H$, and the true problem is whether this is the only obstruction to recovering~$H$.

\subsection{Results}
\label{sec:results}

We will show in theorem~\ref{thm:ideal-reduction} that ideal data (in the sense of definition~\ref{def:ideal} above) is enough to determine the trace free part of the Hamiltonian matrix~$H$ everywhere, provided that certain weighted X-ray transforms are injective.

\begin{quote}
{\bfseries Theorem~\ref{thm:ideal-reduction}:}
Let $\Omega\subset\R^n$, $n\geq2$, be a convex bounded smooth domain and write $M=\bar\Omega$.
Suppose the sets $\initial,\final\subset\C^N$, $N\geq2$, give ideal data.
Let $H\colon M\to\C^{N\times N}$ be continuous and pointwise hermitean.
Assume that we know~$\abs{a^*U^\gamma_Hb}$ for all $a\in\final$, $b\in\initial$ and every line $\gamma\colon[0,T]\to\bar\Omega$ through~$\Omega$.
If~$\rt_W$ is injective for any continuous weight $W\colon SM\to SU(N^2)$, then the data uniquely determines the trace free part of~$H$.
\end{quote}

Here~$\rt_W$ is the X-ray transform with weight~$W$ (see section~\ref{sec:wxrt-intro} for a definition) and $SM=M\times S^{n-1}$ denotes the sphere bundle of $M$.

In theorem~\ref{thm:ideal-special} we will consider the special case where $H(x)=H_0(x)+f(x)G(x)$, where~$H_0$ and~$G$ are known matrix functions and~$f$ is an unknown scalar function; the result is that ideal data determines~$f$, provided that~$G$ is nowhere a multiple of the identity.

\begin{quote}
{\bfseries Theorem~\ref{thm:ideal-special}:}
Let $\Omega\subset\R^n$, $n\geq3$, be a strictly convex bounded smooth domain and write $M=\bar\Omega$.
Suppose the sets $\initial,\final\subset\C^N$, $N\geq2$, give ideal data.
Fix any $\alpha>0$.
Let $H_0,G\colon M\to\C^{N\times N}$ be $C^{1,\alpha}$-smooth and pointwise hermitean.
Let $H(x)=H_0(x)+f(x)G(x)$ for a function $f\in C^{1,\alpha}(M)$.
Assume that we know $\abs{a^*U^\gamma_Hb}$ for all $a\in\final$, $b\in\initial$ and every line $\gamma\colon[0,T]\to\bar\Omega$ through~$\Omega$.
If~$H_0$ and~$G$ are known and~$G$ is nowhere a multiple of the identity, then the data uniquely determines the function~$f$.
\end{quote}

We will discuss the linearized problem in section~\ref{sec:lin}.

In section~\ref{sec:unordered} we will introduce an unordered version of our main problem~\ref{prob:main}.
We emphasize that the unordered problem~\ref{prob:unordered} is neither a generalization nor a specialization of problem~\ref{prob:main}.
For this problem we will only be able to solve a linearized version; the solution of the unordered problem linearized at a constant background is given in theorem~\ref{thm:lin}.

\begin{quote}
{\bfseries Theorem~\ref{thm:lin}:}
Let $\Omega\subset\R^n$, $n\geq2$, be a bounded domain.
Let $H\in C(\bar\Omega,\C^{N\times N})$ take values in hermitean matrices and let~$H_0$ be a fixed hermitean matrix.
Suppose that for every line $\gamma\colon[0,T]\to\bar\Omega$ with endpoints on~$\partial\Omega$ we know
\begin{equation}
e^{iTH_0}\der\exp_{-iTH_0}i\int_0^TH(\gamma(t))\der t
\end{equation}
up to an unknown multiple of the identity matrix.
Then we can recover the trace free part of~$H(x)$ for every $x\in\Omega$.
\end{quote}

Finally, we will analyze the X-ray transform problem with matrix weights that arises in theorem~\ref{thm:ideal-reduction}.
We will show in theorem~\ref{thm:wrt} that in Euclidean domains of dimension three or higher a continuous $\C^N$\mbox{-}valued function is determined by its weighted integrals over all lines when the weight is an invertible~$\C^{N\times N}$ matrix depending on position and direction in a sufficiently continuous way.
In fact, we will prove a local injectivity result for this weighted transform; this is stated separately in theorem~\ref{thm:wrt-local}.

\begin{quote}
{\bfseries Theorem~\ref{thm:wrt}:}
Let~$M$ be the closure of a bounded convex domain in~$\R^n$, $n\geq3$, and let $\beta>0$.
Let $W\in C^\beta(M,\Lip(S^{n-1},\C^{N\times N}))$ be invertible at every point on $SM=M\times S^{n-1}$.
(This regularity assumption is true, in particular, if $W\in C^{1,\beta}(SM,\C^{N\times N})$.)
If $f\in C(M,\C^N)$ satisfies $\rt_Wf=0$, then $f=0$.
\end{quote}

Combining theorems~\ref{thm:ideal-reduction} and~\ref{thm:wrt}, we obtain the following theorem which is our main answer to problem~\ref{prob:main}.

\begin{theorem}
\label{thm:main}
Let~$\Omega$ be a bounded convex domain in~$\R^n$, $n\geq3$, and let $\beta>0$.
Suppose $H,\tilde H\in C^{1,\beta}(\bar\Omega,\C^{N\times N})$.
Suppose the sets $\initial,\final\subset\C^N$, $N\geq2$, give ideal data.
The two Hamiltonians $H$ and $\tilde H$ give the same data in the sense of problem~\ref{prob:main} if and only if $\tilde H=H+\phi I$ for some scalar function $\phi\colon\Omega\to\R$.
\end{theorem}

\section{Physical discussion}
\label{sec:phys}

More detailed physical remarks and background in addition to section~\ref{sec:qm-intro} are in order, and we give them in this section.
A reader who is interested in solving the mathematical problem posed above instead of justifying it physically may wish to skip this section.

In our model a quantum mechanical particle has well defined position and momentum at all times.
This situation is in fact forbidden by the uncertainty principle but it is an excellent approximation.
In a more careful model particles would have to be modeled by wave packets which are nontrivially distributed in position and momentum.
We have (implicitly) assumed that uncertainty in momentum is very small in comparison to the mean value of momentum.
Therefore momentum of a particle as a single vector is a relatively well defined concept, and if momentum is large, even very small relative uncertainty allows for substantial absolute uncertainty.
The necessary amount of uncertainty in position may be so small that a particle can essentially be considered point-like.
Well defined position and momentum (which could also be informally described as a plane wave supported at a point) are a fairly good approximation if momentum is large and the background medium does not have features on shorter length scales than position uncertainty.
When the Earth is imaged by neutrinos, these assumptions are easily met.
(We will also discuss momentum dependent Hamiltonians in section~\ref{sec:tensor}.)

Our mathematical model is only relevant in physical settings where the particles are ultrarelativistic and experience very little scattering or absorption.
Such particles can be reasonably approximated by point-like particles with a quantum mechanical state as modeled in section~\ref{sec:qm-intro}.
This is a common leading order approximation in particle physics.

The original motivation behind this article --- to which the model is not restricted --- is to image the Earth by sending rays through it.
Most particles experience so much absorption or scattering that practically no signal makes it through the planet.
For neutrinos, however, absorption and scattering are very weak and quite difficult to measure.
Intensity of a neutrino beam is not very sensitive to medium properties, so X-ray type imaging is difficult.
On the other hand, a quantum mechanical phenomenon called neutrino oscillation is sensitive to electron density in the medium.
This calls for quantum mechanical measurements of the type described in section~\ref{sec:qm-intro}.
The fact that quantum mechanical evolution is sensitive to something but intensity is not shows that the underlying physical phenomenon is quantum mechanical interference.
Such interference is only possible for a coherent superposition of quantum states, whence the term ``coherent quantum tomography''.
Neutrino oscillation will be discussed in much more detail in section~\ref{sec:nu} below.

The observation that we could --- in principle at least --- see more things with neutrinos than other particles due to the weakness of their interactions is not a new one.
The Universe became transparent to photons at the age of about $380000$ years and optical observations do not allow to see a younger Universe directly.
For neutrinos the corresponding age of the Universe is about one second, whence neutrinos ``see'' significantly older structures than photons do.
The problem is that neutrinos are much, much more difficult to measure than photons; more information is available but it is in a form difficult to measure.
For more about neutrino astrophysics, see~\cite{GK:nu-book,L:galaxy-formation,TM:nu-primordial-anisotropy}.
This is the case for imaging the Earth with neutrinos, too: current technology is insufficient for very precise measurements but there is no fundamental physical obstruction to seeing new things with neutrinos.

The usual computerized X-ray tomography is based on measuring attenuation of photon intensity.
Photons can be replaced with neutrinos as mentioned in the introduction, and computerized tomography with muons has been proposed in~\cite{ST:muon-CT}.
We emphasize that neutrino absorption tomography (as discussed in~\cite{RS:nu-abs-tomography,JRF:nu-abs-tomography,RJF:nu-abs-tomography,W:nu-tomography}) is significantly different from neutrino oscillation tomography discussed in this article.
For physical articles on neutrino oscillation tomography, we refer to~\cite{NJT:nu-tomography-earth,ML:nu-tomography,LOTW:sn-nu-tomography,W:nu-tomography}.
Neutrino tomography could complement existing geophysical imaging methods as discussed in the cited articles.

Of other existing imaging methods based on coherence we mention optical coherence tomography~\cite{S:oct-review} and quantum optical coherence tomography~\cite{TSWS:qoct-review}.
There are also other kinds of inherently quantum mechanical imaging modalities which often go by the name quantum tomography (see eg.~\cite{SKMRBAFW:quantum-measurement,GLFBE:quantum-tomography,MRL:quantum-process-tomography}).

Our imaging method is also related --- from a physical point of view --- to inverse scattering problems for the Schr\"odinger equation (see eg.~\cite{N:schrodinger-scattering,ER:schrodinger-scattering}).
We study quantum mechanical particles that obey the Schr\"odinger equation, but scattering is not a very significant phenomenon for the present imaging problem.
Moreover, in our problem quantum mechanical time evolution takes place in a finite dimensional state space instead of a Hilbert space of functions.

\subsection{Neutrino oscillation}
\label{sec:nu}

For an introduction to neutrino physics and in neutrino oscillation, we refer to~\cite{GK:nu-book,G:flavor-qm-osc,KGP:nu-booklet,I:nu-bsc,TTTY:wave-packet,AS:nu-paradox,MA:solar-nu,BFP:nu-osc-uncertainty}.
We restrict our attention mostly to a quantum mechanical description of neutrino oscillation; quantum field theoretical descriptions can be found, for example, in~\cite{GK:nu-book,AS:nu-paradox}.

There are three kinds of neutrinos ($\nu_e,\nu_\mu,\nu_\tau$) and the state space (excluding motion) of a neutrino is thus a three dimensional complex space.
We think of this space as~$\C^3$ and identify the aforementioned neutrino types with the standard basis vectors $e_1,e_2,e_3$.
This is the so-called flavour basis, and the corresponding basis states ($\nu_e,\nu_\mu,\nu_\tau$) are known as flavour neutrinos.
A neutrino state is a superposition of different flavour neutrinos.

Neutrinos are typically produced as flavour neutrinos, so in our quantum mechanical model $\initial\subset\{e_1,e_2,e_3\}$.
Neutrinos are typically also observed as flavour neutrinos; if all flavours can be measured, we have $\final=\{e_1,e_2,e_3\}$.
Flavour basis is the natural choice if we only consider creation and annihilation in weak charged current reactions (mediated by~$W^\pm$ bosons); if weak neutral current processes are included, the picture is more complicated.

It is sometimes convenient to study neutrino physics in a two-flavour model and it is possible that there are more than three neutrino flavours (the additional ones are known as sterile neutrinos).
If there are sterile neutrinos, they cannot be produced or observed directly, so the basis vectors $e_4,e_5,\dots$ will not appear in~$\initial$ and~$\final$.

Let us then discuss the Hamiltonian related to neutrinos.
There is a hermitean~$3\times3$ matrix~$M$ describing neutrino masses.
There is a matrix~$U$ diagonalizing~$M$, so that
\begin{equation}
U^*MU
=
\begin{pmatrix}
m_1&0&0\\
0&m_2&0\\
0&0&m_3
\end{pmatrix}
,
\end{equation}
where $m_1,m_2,m_3$ are the neutrino masses.
The eigenstates of the matrix~$M$ are known as mass states.
It is of crucial importance in neutrino physics that the mass states are not the same as the flavour states.
The flavour neutrinos do not have well-defined masses, but suitable linear combinations of them (the mass states) do.
The relation of flavour and mass states is expressed by the lepton mixing matrix, the Pontecorvo--Maki--Nakagawa--Sakata matrix~$U$.

Suppose the neutrino is created at a momentum\footnote{One can also assume equal energies instead of equal momenta for the mass states. These assumptions are different, but lead to the same results. For a discussion of the assumptions on the initial state, see~\cite{I:nu-bsc,AS:nu-paradox}.} $p>0$ which is much larger than any of the masses.
If the neutrino had mass $m\in\R$, then its energy would be $E=\sqrt{p^2+m^2}\approx p+\frac{m^2}{2p}$.
Thus in a good approximation we may write the Hamiltonian matrix (recall that the Hamiltonian is the energy operator in quantum mechanics) as $H=pI+\frac1{2p}M^2$, where~$M$ is the mass matrix introduced above.
Since multiples of the identity are irrelevant for the time evolution (they only introduce a total phase), we can consider the Hamiltonian
\begin{equation}
\label{eq:nu-H0}
H_0
=
\frac1{2p}
U
\begin{pmatrix}
m_1^2&0&0\\
0&m_2^2&0\\
0&0&m_3^2
\end{pmatrix}
U^*.
\end{equation}
This simple Hamiltonian gives rise to a remarkable physical phenomenon: neutrino oscillation.
If a neutrino is produced in a flavour state, the probability that it is observed in another flavour state (known as transition probability) oscillates in time.
It is evident from equation~\eqref{eq:nu-H0} that the key parameters for neutrino oscillations are the PMNS matrix~$U$ and the differences of the squared neutrino masses.
The momentum~$p$ in~\eqref{eq:nu-H0} can be replaced with the energy~$E$, as these are almost the same for ultrarelativistic particles.

This oscillation is produced by quantum mechanical interference of different neutrino mass states.
Interference is only possible if the different mass states are coherent and in the same place.
Different masses give rise to slightly different speeds, and different mass states drift apart slowly, whence coherence is lost and neutrino oscillation comes to an end.
Transition probabilities stop oscillating, but the asymptotic probabilities are not trivial (not only zero or one).
For coherence and decoherence in neutrino oscillations, we refer to~\cite[Section~4]{I:nu-bsc}.
(We wish to point out that lost coherence can be recovered at detection; see~\cite[Section~5.4]{AS:nu-paradox}.)

Neutrino oscillation is a kinematic effect; it requires no interactions.
The neutrino oscillation described hitherto is the kind that takes place in vacuum.
Interactions with the medium change the Hamiltonian of~\eqref{eq:nu-H0} and also the oscillation.
It is via these changes that we hope to determine properties of the medium using neutrino oscillation.

If the neutrinos propagate in matter with electron number density~$N_e$, we have to add the potential term
\begin{equation}
\label{eq:cc-potential}
2\sqrt2 E G_F N_e
\begin{pmatrix}
1&0&0\\
0&0&0\\
0&0&0
\end{pmatrix}
\end{equation}
(with a minus sign for antineutrinos) to the Hamiltonian of~\eqref{eq:nu-H0}, neglecting again multiples of the identity matrix.
Here~$E$ is the neutrino energy and~$G_F$ is the Fermi constant.
This particular form of the potential is due to only electron neutrinos interacting with electrons in the medium via weak charged current; the neutral current potentials of electrons and protons cancel each other and the one of neutrons is a multiple of identity.
For details of neutrino oscillation in matter, we refer to~\cite[Section~9.2]{GK:nu-book}.

Because of the form of the potential in~\eqref{eq:cc-potential}, neutrino oscillations provide a means of measuring electron density as a function of position.
Replacing momentum with energy, we arrive at the total Hamiltonian
\begin{equation}
H(x)
=
\frac1{2E}
U
\begin{pmatrix}
m_1^2&0&0\\
0&m_2^2&0\\
0&0&m_3^2
\end{pmatrix}
U^*
+
Ef(x)
\begin{pmatrix}
1&0&0\\
0&0&0\\
0&0&0
\end{pmatrix}
,
\end{equation}
where~$f$ is a scaled electron density.
The problem is to recover the function~$f$ from the data described in section~\ref{sec:qm-intro}, possibly with measurements for several energies~$E$.
We emphasize that ideal data is not available, so theorem~\ref{thm:ideal-special} is not applicable.

\subsection{Tensor Hamiltonians}
\label{sec:tensor}

So far we have only considered Hamiltonians~$H(x)$ as functions of $x\in M$.
It may happen, however, that the Hamiltonian also depends on direction, so that~$H(x,v)$ is a function of $(x,v)\in SM$.
If there is no a priori knowledge on how~$H$ depends on~$v$, problem~\ref{prob:main} is hopeless, as the time evolution operators along intersecting lines have nothing to do with each other.

A natural assumption is that~$H(x,v)$ is a polynomial in~$v$.
From the point of view of integral geometry, this corresponds to replacing scalar fields by (sums of) tensor fields, the order of the polynomial being the order of the tensor field.
A particularly simple case to analyze is a first order polynomial (corresponding to a sum of a function and a one-form).
For more details on tensor tomography, we refer to~\cite{PSU:tensor-survey,S:tensor-book}.

The polynomial assumption is also natural physically, and we illustrate this with examples related to particles in an electromagnetic field.
The Dirac equation in~$\R^3$ for a fermion of charge~$q$ and mass~$m$ can be written as $i\partial_t\Psi=H\Psi$ with the Hamiltonian\footnote{This form can be obtained from the standard Dirac equation by describing the coupling to the photon field by the gauge covariant four-gradient $\partial_\mu+iqA_\mu$ and then solving for the time derivative.}
\begin{equation}
\label{eq:dirac-H}
H(x,p)
=
q\phi(x)I+m\gamma^0+(p-qA(x))\cdot\alpha.
\end{equation}
The state $\Psi\in\C^4$ is the Dirac spinor.
For a vector $V\in\R^3$, $V\cdot\alpha=V_1\alpha^1+V_2\alpha^2+V_3\alpha^3$, the~$4\times4$ matrices $\alpha^1,\alpha^2,\alpha^3$ are related to the Dirac gamma matrices by $\alpha^i=\gamma^0\gamma^i$, $A$ is the magnetic vector potential and~$\phi$ is the electric scalar potential.
The Hamiltonian depends on the momentum~$p$, which can be turned into dependence on direction and energy using the identity $p=Ev$.
The nonrelativistic counterpart of the Dirac equation is the Pauli equation with the Hamiltonian
\begin{equation}
\label{eq:pauli-H}
H(x,p)
=
\left(\frac1{2m}(p-qA(x))^2+q\phi(x)\right)I+\frac{q}{2m}B(x)\cdot\sigma
\end{equation}
acting on a spinor $\Psi\in\C^2$, where $B=\nabla\times A$ is the magnetic field, $B\cdot\sigma=B_1\sigma_1+B_2\sigma_2+B_3\sigma_3$ and the~$\sigma$s are the Pauli matrices.

There are several possible choices for the matrices appearing in~\eqref{eq:dirac-H} and~\eqref{eq:pauli-H}.
For completeness, we give one here:
\begin{equation}
\begin{split}
\sigma_1
&=
\begin{pmatrix}
0&1\\
1&0
\end{pmatrix},
\\
\sigma_2
&=
\begin{pmatrix}
0&-i\\
i&0
\end{pmatrix},
\\
\sigma_3
&=
\begin{pmatrix}
1&0\\
0&-1
\end{pmatrix},
\\
\gamma^0
&=
\sigma_3\otimes I_2
\quad\text{and}
\\
\gamma^i
&=
\begin{pmatrix}
0&1\\
-1&0
\end{pmatrix}
\otimes\sigma_i
\quad\text{for }i=1,2,3,
\end{split}
\end{equation}
where~$I_2$ is the~$2\times2$ identity matrix.
The mathematically oriented reader can neglect positioning of indices (upper or lower); we used some upper indices with the Dirac equation and Dirac matrices to be consistent with notation in particle physics.

The natural problem related to these equations is to recover the electromagnetic potentials~$\phi$ and~$A$ from the data described in section~\ref{sec:qm-intro}.
The weighted ray transform described in section~\ref{sec:wxrt-intro} has to be generalized further: we have to consider matrix valued tensor fields, or functions $SM\to\C^{N\times N}$ that are polynomial in the fiber variable~$v$.
As with all tensor tomography problems, we expect that there is some gauge freedom.
But this is to be expected on physical grounds as well, since one can only ever measure a (magnetic or other) potential up to gauge.

We state this tensor tomography problem separately:

\begin{problem}
\label{prob:tensor}
What happens in problem~\ref{prob:matrix-wxrt} if the unknown functions are assumed to be matrix valued tensor fields on~$SM$?
What if we assume the weight to satisfy $XW=iWQ$ for a matrix valued tensor field~$Q$ (cf.\ remark~\ref{rmk:att})?
\end{problem}

%

\section{Gauge conditions for measurements}
\label{sec:gauge}

Given $\initial,\final\subset\C^N$, how much about $U\in U(N)$ can be inferred from the knowledge of~$\abs{a^*Ub}$ for all $a\in\final$ and $b\in\initial$?
The best we can hope for is knowing~$U$ up to a factor~$e^{i\phi}$, $\phi\in\R$, which means ideal data.

\begin{lemma}
\label{lma:unitary-phase}
If~$\initial$ and~$\final$ are both the unit sphere in~$\C^N$, they give ideal data.
\end{lemma}

\begin{proof}
We need to show that if for two matrices $U,V\in U(N)$ we have $\abs{a^*Ub}=\abs{a^*Vb}$ for all $a,b\in\C^N$ with $\abs{a}=\abs{b}=1$, then $U=\lambda V$ for some $\lambda\in\C$, $\abs{\lambda}=1$.
By a scaling argument we can drop the assumption $\abs{a}=\abs{b}=1$.

It follows from the condition that for any $a\in\C^N$ a vector is orthogonal to~$Ua$ if and only if it is orthogonal to~$Va$.
Therefore there is $\lambda_a\in\C$ so that $Va=\lambda_aUa$.
We wish to show that all the constants~$\lambda_a$, $a\in\C^N\setminus\{0\}$, are in fact equal.

If $a,b\in\C^N$ are linearly dependent, then by linearity $\lambda_a=\lambda_b$.
If $a,b\in\C^N$ are linearly independent, then
\begin{equation}
\lambda_aUa+\lambda_bUb
=
Va+Vb
=
V(a+b)
=
\lambda_{a+b}U(a+b)
=
\lambda_{a+b}Ua+\lambda_{a+b}Ub,
\end{equation}
from which we get
\begin{equation}
(\lambda_{a+b}-\lambda_a)Ua
=
(\lambda_b-\lambda_{a+b})Ub.
\end{equation}
Since~$Ua$ and~$Ub$ are linearly independent, we have $\lambda_a=\lambda_{a+b}=\lambda_b$.
\end{proof}

\begin{problem}
\label{prob:ideal}
How large do the sets $\initial,\final\subset\C^N$ have to be to give ideal data?
The assumption of the previous lemma looks like an overkill.
\end{problem}

The following lemma shows that if we assume the Hamiltonian to be trace free (the trace is invisible from the data anyway), then in the case of ideal data we can recover the time evolution operator~$U^\gamma_H$ for every line~$\gamma$ without any gauge ambiguity.
After this reduction we may restrict ourselves to trace free Hamiltonians, special unitary time evolution and most importantly an inverse problem with no obvious gauge invariance.

\begin{lemma}
\label{lma:SU-trace}
Let $\Omega\subset\R^n$, $n\geq2$, be a convex, bounded domain and let~$M$ denote the space of continuous hermitean matrix valued functions $\bar\Omega\to\C^{N\times N}$.
Fix the sets $\initial,\final\subset\C^N$ so that they give ideal data.
The following are equivalent:
\begin{enumerate}
\item From the knowledge of~$\abs{a^*U^\gamma_Hb}$ for all states $a\in\final$, $b\in\initial$ and every line $\gamma\colon[0,T]\to\bar\Omega$ through~$\Omega$ one can recover the function $H\in M$ up to a function~$f(x)I$ for a continuous function $f\colon\bar\Omega\to\R$.
\item From the knowledge of~$U^\gamma_H$ for every line $\gamma\colon[0,T]\to\bar\Omega$ through~$\Omega$ one can recover the trace free function $H\in M$.
\end{enumerate}
\end{lemma}

\begin{proof}
1$\implies$2:
If we know $U^\gamma_H=U^\gamma_H(T,0)$ for every line $\gamma\colon[0,T]\to\bar\Omega$, then by assumption we can recover~$H$ up to multiples of the identity.
But since we know that~$H$ is trace free, there is no ambiguity.

2$\implies$1:
We write $H(x)=A(x)+f(x)I$, where $\tr(A(x))=0$ for every $x\in\bar\Omega$.
Now for any line $\gamma\colon[0,T]\to\bar\Omega$,
\begin{equation}
U^\gamma_H
=
\exp\left(-i\int_0^Tf(\gamma(t))\der t\right)
U^\gamma_A.
\end{equation}
Since we have ideal data, we know~$U^\gamma_H$ up to phase (by lemma~\ref{lma:unitary-phase}), and since~$A$ is traceless, the matrix~$U^\gamma_A$ is special unitary.
Therefore we know the matrix~$U^\gamma_A$ up to multiplication by an element of the group $G_N\subset U(1)$ of $N$th roots of unity.

We know that for very short lines~$\gamma$ the matrix~$U^\gamma_A$ is very close to the identity, so by finiteness of~$G_N$ we can recover the matrix~$U^\gamma_A$ itself (without any gauge ambiguity) for sufficiently short lines.
By convexity of the domain every line with endpoints on the boundary can be continuously deformed into an arbitrarily short one; it then follows from continuity of~$H$ and finiteness of~$G_N$ that we can in fact recover~$U^\gamma_A$ for all~$\gamma$.
By assumption, this information is enough to determine~$A$, the trace free part of~$H$.
\end{proof}

The argument in the proof of the lemma above should also work with piecewise continuity, but we refrain from pursuing optimal regularity.

It is a simple exercise to show that if the sets $\initial,\final\subset\C^N$ give ideal data and $U,V\in U(N)$, then also the pair $(U\initial,V\final)$ gives ideal data.
From this it follows that if the tomography problem can be solved on a domain $\omega\subset\R^n$ and the trace free part of the Hamiltonian is known in $\Omega\setminus\omega$ for some larger set $\Omega\subset\R^n$, then the problem can also be solved in~$\Omega$.
If the particle travels on a concatenation of several curves, the total time evolution operator is simply the composition of time evolution operators over each segment.
Without ideal data there is more gauge freedom than a simple phase, and it is less obvious to find such layer stripping results.

\begin{lemma}
\label{lma:2D-gauge}
Suppose $N=2$ and let~$e_1$ and~$e_2$ be the basis vectors of~$\C^2$.
If $\initial=\{e_1\}$ and $\final=\{e_1,e_2\}$ (or vice versa), then the following are equivalent for two matrices $U,V\in U(2)$:
\begin{enumerate}
\item $\abs{a^*Ub}=\abs{a^*Vb}$ for all $a\in\final$ and $b\in\initial$.
\item There are angles $\alpha,\beta,\theta\in\R$ so that
\begin{equation}
V
=
e^{i\theta}
\begin{pmatrix}
e^{i\alpha} & 0 \\
0 & e^{-i\alpha}
\end{pmatrix}
U
\begin{pmatrix}
e^{i\beta} & 0 \\
0 & e^{-i\beta}
\end{pmatrix}.
\end{equation}
\end{enumerate}
\end{lemma}

\begin{proof}
2$\implies$1:
Follows from a simple calculation.

1$\implies$2:
Suppose $\abs{a^*Ue_1}=\abs{a^*Ve_1}$ for both $a\in\{e_1,e_2\}$.
Since the columns of~$U$ and those of~$V$ are orthonormal and the dimension is two, we have also $\abs{a^*Ue_2}=\abs{a^*Ve_2}$ for both $a\in\{e_1,e_2\}$.
Let
\begin{equation}
U
=
\begin{pmatrix}
a & b \\
c & d
\end{pmatrix}
\end{equation}
and
\begin{equation}
V
=
\begin{pmatrix}
\kappa a & \lambda b \\
\mu c & \nu d
\end{pmatrix}
\end{equation}
for complex numbers $\kappa,\lambda,\mu,\nu$ of unit norm.

If $c=0$, then $b=0$ and $\abs{a}=\abs{d}=1$ and the claim is obvious.
Similarly, if $a=0$ (and thus $d=0$), the claim follows easily.

Suppose then that $a,b,c,d\neq0$.
From the orthogonality relations
\begin{equation}
\begin{cases}
a\bar c+b\bar d=0 \\
\kappa\bar\mu a\bar c+\lambda\bar\nu b\bar d=0
\end{cases}
\end{equation}
we obtain $\lambda\mu=\kappa\nu$.
We want to choose the angles $\alpha,\beta,\theta$ so that
\begin{equation}
\begin{split}
\kappa&=e^{i(\theta+\alpha+\beta)},\\
\lambda&=e^{i(\theta+\alpha-\beta)},\\
\mu&=e^{i(\theta-\alpha+\beta)}\text{ and}\\
\nu&=e^{i(\theta-\alpha-\beta)}.
\end{split}
\end{equation}
It is easy to verify that this is indeed possible due to the constraint $\lambda\mu=\kappa\nu$.
\end{proof}

\begin{remark}
It is natural to ask if lemma~\ref{lma:2D-gauge} holds for $N>2$ assuming $\initial=\final=\{e_1,\dots,e_N\}$.
The analogue of implication 2$\implies$1 is easy to verify, but the converse is false.
Condition~2 of the lemma describes a symmetry group (that fixes norms of matrix elements) of dimension $2N-1$.

We may linearize (w.r.t.\ the phase of each matrix element) to find the dimension of the tangent space of the full symmetry group and thus the dimension of that group.
To simplify matters, we assume each matrix element to be nonzero.
Experiments show that the group of symmetries preserving the norms of elements of a generic~$3\times3$ unitary matrix has dimension 6, not~5.
We have been unable to identify the group itself in a spirit similar to condition~2 of lemma~\ref{lma:2D-gauge}.
\end{remark}

The following problem needs to be solved before one can study our main problem for $N\geq3$ with $\initial=\final=\{e_1,\dots,e_N\}$ the standard basis of~$\C^N$.
This problem with $N=3$ is relevant for neutrino oscillations, see section~\ref{sec:nu}.

\begin{problem}
\label{prob:gauge}
Given a unitary matrix $U\in U(N)$, $N\geq3$, describe the set
\begin{equation}
\{V\in U(N);\abs{V_{ij}}=\abs{U_{ij}}\text{ for all }i,j\}.
\end{equation}
Is there a description as simple as that of lemma~\ref{lma:2D-gauge} for $N=2$?
\end{problem}

\section{Pseudolinearization}
\label{sec:pseudolin}

Here we introduce a so-called pseudolinearization that reduces our nonlinear problem to a linear one with unknown parameters.
It is enough to know some general characteristics of these parameters (certain weight functions) to show injectivity of the linear problem.
It also gives rise to an iterative reconstruction scheme which we shall discuss after theorem~\ref{thm:ideal-reduction}.
A similar idea has been previously used for rigidity problems~\cite{SU:rigidity98,SUV:partial-bdy-rig}.

We will phrase our results in the setting of continuous unknown functions.
If we assume more regularity, we can assume more regularity in the weights of the corresponding ray transform problems.

\subsection{Matrix-weighted X-ray transform}
\label{sec:wxrt-intro}


Let $\Omega\subset\R^n$, $n\geq2$, be a bounded smooth domain, and denote $M=\bar\Omega$.
Denote by $SM=S^{n-1}\times M$ the unit sphere bundle of~$M$.
Fix an integer $N\geq2$ and a continuous function $W\colon SM\to SU(N)$.
For a continuous function $f\colon M\to\C^N$, define the $W$\mbox{-}weighted X-ray transform~$\rt_Wf$ of~$f$ by letting
\begin{equation}
\rt_Wf(\gamma)
=
\int_0^TW(\gamma(t),\dot\gamma(t))f(t)\der t
\end{equation}
for every unit speed line $\gamma\colon[0,T]\to M$ with endpoints on~$\partial M$.

As it will shortly turn out, pseudolinearization leads to this integral transform.
Weighted and especially attenuated X-ray transforms have been studied before (see eg.~\cite{F:attenuated-x-ray,MQ:weighted-xrt,B:weighted-radon-noninjective,SU:surface,SF:attenuated-x-ray,BS:attenuated-x-ray,M:x-ray-numerics,M:attenuated-xrt,V:matrix-wxrt}), but our results about X-ray transforms with matrix weights are new to the best of our knowledge.
The matrix weight is, in a sense, an exponential, and could therefore be regarded as an attenuation (see remark~\ref{rmk:att}).
The attenuated X-ray transform with matrix attenuation is related to recovering a matrix valued connection from is parallel transport; for results about connection problems, see~\cite{FU:connection,GPSU:connection,PSU:connection-higgs,N:matrix-radon}.

The (scalar) weights that arise in the proof of theorem~\ref{thm:ideal-special} are not attenuations.
We have results for more general weights in section~\ref{sec:wxrt} and we do not wish to pursue optimal regularity at the expense of clarity, so we will not study the attenuated case further here.

\begin{problem}
\label{prob:matrix-wxrt}
Is the transform~$\rt_W$ injective on continuous ($\C^N$\mbox{-}valued) functions whenever $n,N>1$?
What if we assume the weight~$W$ and the function~$f$ to be smooth?
What if~$M$ is replaced with a compact Riemannian manifold with boundary?
\end{problem}

\subsection{Ideal data}
\label{sec:ideal-data}

We are now ready to pseudolinearize our problem.
We assume ideal data in the sense of definition~\ref{def:ideal}.

\begin{theorem}
\label{thm:ideal-reduction}
Let $\Omega\subset\R^n$, $n\geq2$, be a convex bounded smooth domain and write $M=\bar\Omega$.
Suppose the sets $\initial,\final\subset\C^N$, $N\geq2$, give ideal data.
Let $H\colon M\to\C^{N\times N}$ be continuous and pointwise hermitean.
Assume that we know~$\abs{a^*U^\gamma_Hb}$ for all $a\in\final$, $b\in\initial$ and every line $\gamma\colon[0,T]\to\bar\Omega$ through~$\Omega$.
If~$\rt_W$ is injective for any continuous weight $W\colon SM\to SU(N^2)$, then the data uniquely determines the trace free part of~$H$.
\end{theorem}

\begin{proof}
By lemma~\ref{lma:SU-trace} we can assume that~$H$ is trace free and we know~$U^\gamma_H$ for every line~$\gamma$.
We have replaced~$H$ with its trace free part, so we aim to recover the function~$H$ without any gauge restrictions.

Suppose that for continuous, hermitean, trace free matrix functions $H,\tilde H\colon SM\to\C^{N\times N}$ we have $U^\gamma_H=U^\gamma_{\tilde H}$ for every line~$\gamma$ through~$M$.
We wish to show that $H=\tilde H$.

Fix any line $\gamma\colon[0,T]\to M$ with endpoints on~$\partial M$.
Define $V\colon[0,T]\to SU(N)$ by
\begin{equation}
V(t)
=
U_H^\gamma(T,t)
U_{\tilde H}^\gamma(t,0).
\end{equation}
Since $\tr(H)=\tr(\tilde H)\equiv0$, all time evolution operators are in fact special unitary.

Since
\begin{equation}
\label{eq:U-der}
\begin{split}
\Der{t}U_H^\gamma(T,t)
&=
iU_H^\gamma(T,t)H(\gamma(t))
\quad\text{and}\\
\Der{t}U_{\tilde H}^\gamma(t,0)
&=
-i\tilde H(\gamma(t))U_{\tilde H}^\gamma(t,0),
\end{split}
\end{equation}
we have
\begin{equation}
\label{eq:pseudolin1}
\begin{split}
0
&=
U_{\tilde H}^\gamma(T,0)-U_H^\gamma(T,0)
\\&=
V(T)-V(0)
\\&=
\int_0^T\Der{t}V(t)\der t
\\&=
i\int_0^T U_H^\gamma(T,t) (H(\gamma(t))-\tilde H(\gamma(t))) U_{\tilde H}^\gamma(t,0)\der t.
\end{split}
\end{equation}

For every $(x,v)\in SM$, let $\gamma_{x,v}\colon[0,T_{x,v}]\to M$ be the unique maximal line passing through point~$x$ in direction~$v$.
Let $t_{x,v}\in[0,T_{x,v}]$ be such that $\gamma_{x,v}(t_{x,v})=x$.
Define $A,B\colon SM\to SU(N)$ by
\begin{equation}
\label{eq:AB-def}
\begin{split}
A(x,v)&=U_H^{\gamma_{x,v}}(T_{x,v},t_{x,v})
\quad\text{and}\\
B(x,v)&=U_{\tilde H}^{\gamma_{x,v}}(t_{x,v},0).
\end{split}
\end{equation}
These functions are continuous.
Define $f\colon M\to \C^{N\times N}$ by $f(x)=H(x)-\tilde H(x)$ and let $W\colon SM\to GL(\C^{N\times N})$ be defined by
\begin{equation}
W(x,v)(Z)
=
A(x,v)ZB(x,v)
\end{equation}
for every matrix $Z\in\C^{N\times N}$.
If we endow~$\C^{N\times N}$ with the standard inner product, then the linear transformation $Z\mapsto A(x,v)ZB(x,v)$ is a special unitary one since the matrices~$A(x,v)$ and~$B(x,v)$ are special unitary.
(We may think of~$\C^{N\times N}$ as the tensor product~$\C^N\otimes\C^N$ with~$A$ acting on the first component and~$B$ acting on the second one. The tensor product of special unitary matrices is again special unitary.)

Now we have a continuous function $f\colon M\to \C^{N^2}$ and a continuous weight $W\colon SM\to SU(N^2)$.
Equation~\eqref{eq:pseudolin1} reads now $\rt_Wf=0$.
Since such transforms were assumed to be injective, we have $f=0$ and thus $H=\tilde H$.
\end{proof}

Pseudolinearization can also be used for reconstruction.
One first starts with an initial guess~$H_0$.
The data for~$H_0$ and the true Hamiltonian~$H_\infty$ determines an X-ray transform of $H_0-H_\infty$ with weights given by~$H_0$ and~$H_\infty$; cf.~\eqref{eq:pseudolin1}.
One can then approximate $H_0-H_\infty$ by inverting this transform with weights given by~$H_0$ and~$H_0$.
(This means treating the pseudolinearization as if it were the linearization; cf.\ section~\ref{sec:lin}.)
This gives rise to an improved guess~$H_1$ for~$H_\infty$, and one may proceed iteratively to find successive approximations $H_2,H_3,H_4,\dots$ using the same method.
Then hopefully $H_k\to H_\infty$ as $k\to\infty$ in a suitable sense.
This idea been used to solve inverse problems numerically via pseudolinearization in~\cite{CQUZ:pseudolin-numerics1,CQUZ:pseudolin-numerics2}.

\begin{remark}
\label{rmk:att}
The weight that arose in the proof of theorem~\ref{thm:ideal-reduction} is of a special kind.
Identifying~$\C^{N\times N}$ with $\C^N\otimes\C^N$, the weight is $W(x,v)=A(x,v)\otimes B(x,v)^*$.
Let us define the differential operator $X=v\cdot\nabla_x$ acting on functions on~$SM$; this operator is known as the geodesic vector field and it generates the geodesic flow.
Comparing equations~\eqref{eq:AB-def} and~\eqref{eq:U-der}, we observe that $XA=iAH$ and $XB=-iB\tilde H$, whence $XW=iWQ$, where $Q=H\otimes I+I\otimes\tilde H$ (which is hermitean).
The weight can therefore be regarded as an attenuation.
We will prove an injectivity result for weighted X-ray transforms even for very rough weights in section~\ref{sec:wxrt}, and we need not assume that the weight is an attenuation.
A more detailed discussion of X-ray transforms is given in section~\ref{sec:wxrt}.
%
\end{remark}

The following result is based on weighted X-ray transforms with scalar weights.
We present the theorem in the amount of generality allowed for by our result on local weighted X-ray transforms, theorem~\ref{thm:wrt-local}.
Results for such transforms are known also on manifolds~\cite{FSU:general-x-ray,SUV:partial-bdy-rig}, so the result can be easily generalized.

\begin{theorem}
\label{thm:ideal-special}
Let $\Omega\subset\R^n$, $n\geq3$, be a strictly convex bounded smooth domain and write $M=\bar\Omega$.
Suppose the sets $\initial,\final\subset\C^N$, $N\geq2$, give ideal data.
Fix any $\alpha>0$.
Let $H_0,G\colon M\to\C^{N\times N}$ be $C^{1,\alpha}$-smooth and pointwise hermitean.
Let $H(x)=H_0(x)+f(x)G(x)$ for a function $f\in C^{1,\alpha}(M)$.
Assume that we know $\abs{a^*U^\gamma_Hb}$ for all $a\in\final$, $b\in\initial$ and every line $\gamma\colon[0,T]\to\bar\Omega$ through~$\Omega$.
If~$H_0$ and~$G$ are known and~$G$ is nowhere a multiple of the identity, then the data uniquely determines the function~$f$.
\end{theorem}

\begin{proof}
By lemma~\ref{lma:SU-trace} we can assume~$H_0$ and~$G$ to be trace free.
The trace free part of~$G$ does not vanish anywhere.

Suppose two functions $f,\tilde f\in C^{1,\alpha}(M)$ give the same data.
As in~\eqref{eq:pseudolin1}, we get
\begin{equation}
\label{eq:pseudolin2}
\int_0^T (f(\gamma(t))-\tilde f(\gamma(t))) U_H^\gamma(T,t) G(\gamma(t)) U_{\tilde H}^\gamma(t,0)\der t
=
0
\end{equation}
for all lines $\gamma\colon[0,T]\to M$ with endpoints on~$\partial M$.
If the line goes through $x\in M$ and has direction $v\in S^{n-1}$ at time~$t$, let us denote
\begin{equation}
w_{ij}(x,v)=[U_H^\gamma(T,t) G(\gamma(t)) U_{\tilde H}^\gamma(t,0)]_{ij}
\end{equation}
for all indices $1\leq i,j\leq N$.
In this way we can define $w_{ij}(x,v)$ for all $(x,v)\in SM$.

By~\eqref{eq:pseudolin2} know that the scalar function $g=f-\tilde f$ has vanishing weighted X-ray transform for each of the weights~$w_{ij}$.
Since the matrix whose coefficients are~$w_{ij}$ is nonzero, the weights cannot vanish simultaneously.
The weights are also~$C^{1,\alpha}$ on~$SM$, so every point $(x,v)\in SM$ has a neighborhood where one of the weights~$w_{ij}$ is bounded away from zero.
In particular, if $x\in\partial M$ and~$v$ is tangent to~$\partial M$, there is such a neighborhood.
Furthermore, since~$g$ is real valued, its weighted X-ray transform vanishes also with the weight~$\Re(\lambda w_{ij})$ for any $\lambda\in\C$; we can thus always locally find a strictly positive real weight for the X-ray transform.

Now, using theorem~\ref{thm:wrt-local} with $N=1$, we observe that $f-\tilde f=0$ in some neighborhood of~$x$.
We may iterate this argument (a similar layer stripping method was presented in~\cite{UV:local-x-ray}) to show that $f-\tilde f=0$ in all of~$M$.
The method presented in~\cite{UV:local-x-ray} requires a foliation condition that is satisfied by a strictly convex Euclidean domain.
(Instead of iteration, one can also argue by contradiction as in the proof of theorem~\ref{thm:wrt}.)
\end{proof}

\subsection{Linearization}
\label{sec:lin}

Since we have been able to give some answers to the full nonlinear problem~\ref{prob:main}, we will only briefly discuss its linearization.

The linearized problem will lead to a similar weighted X-ray transform problem that arose in the proof of theorem~\ref{thm:ideal-reduction}.
The difference is that the weight will be known (dependent on the reference Hamiltonian at which we linearize) and therefore easier to invert.
For completeness, we prove below that this is indeed the correct linearization.
Linearization may be a relevant approach when the Hamiltonian is known to be close to a known Hamiltonian or when one wishes to track (slow) time dependence of the Hamiltonian.

\begin{proposition}
\label{prop:lin}
Let~$H_s$ be a family of smooth Hamiltonians in a domain $\Omega\subset\R^n$ parameterized smoothly by $s\in(-\eps,\eps)$ for some $\eps>0$.
Let $\gamma\colon[0,T]\to\bar\Omega$ be a line.
Then
\begin{equation}
\Der{s}U_{H_s}^\gamma|_{s=0}
=
i\int_0^TU_{H_0}^\gamma(T,t)\Der{s}H_s(\gamma(t))|_{s=0}U_{H_0}^\gamma(t,0)\der t.
\end{equation}
\end{proposition}

\begin{proof}
We can see from equation~\eqref{eq:pseudolin1} that for any~$s$ we have
\begin{equation}
\label{eq:vv4}
U_{H_s}^\gamma-U_{H_0}^\gamma
=
i\int_0^TU_{H_0}^\gamma(T,t)(H_0(\gamma(t))-H_s(\gamma(t)))U_{H_s}^\gamma(t,0)\der t.
\end{equation}
Let us denote $H'(x)=\Der{s}H_s(x)|_{s=0}$.
Since $H_0(\gamma(t))-H_s(\gamma(t))=sH'(\gamma(t))+\order(s)$ and $U_{H_s}^\gamma(t,0)=U_{H_0}^\gamma(t,0)+\order(1)$ as $s\to0$ (follows from~\eqref{eq:vv4}), we have
\begin{equation}
U_{H_s}^\gamma-U_{H_0}^\gamma
=
is\int_0^TU_{H_0}^\gamma(T,t)H'(\gamma(t))U_{H_0}^\gamma(t,0)\der t+\order(s)
\end{equation}
as claimed.
\end{proof}


\section{Unordered time evolution}
\label{sec:unordered}

In section~\ref{sec:qm-intro} we introduced the time evolution operator and mentioned that it could be expressed as a time ordered exponential.
We will now investigate what happens if we replace the time ordered exponential with the usual matrix exponential.
This problem is no longer quantum mechanically relevant, but it makes an interesting mathematical inverse problem all the same.

If the Hamiltonian~$H$ satisfies $[H(x),H(y)]=0$ for all $x,y\in\Omega$, then the time ordered and the usual unordered exponentials of all integrals of~$H$ coincide.
We do not assume this commutator property in this section, and therefore the time ordered and regular matrix exponentials are different.

We have a set of initial states $\initial\subset\C^N$ and a set of final states $\final\subset\C^N$ and for all lines $\gamma\colon(0,T)\to\Omega$ with endpoints on~$\partial\Omega$ we know the inner products
\begin{equation}
\abs{\Phi^*\exp\left(-i\int_0^TH(\gamma(t))\der t\right)\Psi}
\end{equation}
for all $\Phi\in\final$ and $\Psi\in\final$.
The only difference to problem~\ref{prob:main} is that~$\Texp$ is replaced with~$\exp$.

In the case of ideal data this means that the matrix
\begin{equation}
\label{eq:vv2}
\exp\left(-i\int_0^TH(\gamma(t))\der t\right)
\end{equation}
can be recovered up to a phase.
For ideal data, the unordered problem is now this:

\begin{problem}
\label{prob:unordered}
Let $\Omega\subset\R^n$ be a domain.
If for all line segments $\gamma\colon[0,T]\to\bar\Omega$ with endpoints on~$\partial\Omega$ we know the matrix~\eqref{eq:vv2}, how much can we infer about the Hamiltonian~$H$?
\end{problem}

This problem quickly leads to this auxiliary problem:

\begin{problem}
\label{prob:matrix-exp}
Let $A,B$ be two hermitean~$N\times N$ matrices.
When do we have $e^{iA}=e^{iB}$?
Can we classify the pairs of matrices that have the same exponential?
\end{problem}

We give a solution to this problem in proposition~\ref{prop:matrix-exp} below, but, unfortunately, this result is difficult to work with.
Therefore instead of the full problem we treat the linearization.
Dropping time ordering from the exponential makes a huge difference in the linearization, as we shall see below.
The answer to problem~\ref{prob:unordered}, linearized at any constant Hamiltonian, is given in theorem~\ref{thm:lin} below.


\subsection{When do two matrices have the same exponential?\hspace{-100cm}}\hspace{-.3em}\hspace{100cm}
Fix some integer $N\geq1$.
For an~$N\times N$ matrix~$A$, let us denote its eigenspace corresponding to the eigenvalue $\lambda\in\C$ by
\begin{equation}
E^A_\lambda
=
\ker(A-\lambda I).
\end{equation}
Let us also define the ``periodic eigenspaces''
\begin{equation}
F^A_\lambda
=
\bigoplus_{k\in\Z}E^A_{\lambda+2\pi k}.
\end{equation}
With these spaces we can formulate a simple answer to problem~\ref{prob:matrix-exp}.

\begin{proposition}
\label{prop:matrix-exp}
Let $N\geq1$ be an integer and~$A$ and~$B$ two hermitean~$N\times N$ matrices.
Then the following are equivalent:
\begin{enumerate}
\item $e^{iA}=e^{iB}$,
\item $F^A_\lambda=F^B_\lambda$ for all $\lambda\in\R$.
\end{enumerate}
\end{proposition}

\begin{proof}
2$\implies$1:
Let $0\leq\lambda_1<\lambda_2<\dots<\lambda_n<2\pi$ be the numbers for which $\dim(F^A_{\lambda_k})>0$ and denote the corresponding spaces by $S_k=F^A_{\lambda_k}=F^B_{\lambda_k}$.
Now~$\C^N$ is an orthogonal direct sum of the spaces $S_1,\dots,S_n$ and the matrices~$A$ and~$B$ are block diagonal in this decomposition (that is, $A(S_k),B(S_k)\subset S_k$ for all~$k$).
Let~$A_k$, $1\leq k\leq n$, denote the corresponding block of~$A$, and similarly~$B_k$, $U_k$ and~$V_k$ for~$B$, $U=e^{iA}$ and $V=e^{iB}$.

Considering $k$th block only, we have $e^{iA_k}=U_k$ and $e^{iB_k}=V_k$.
All eigenvalues of~$A_k$ are in $\lambda_k+2\pi\Z$, so the only eigenvalue of~$U_k$ is~$e^{i\lambda_k}$.
Similarly, the only eigenvalue of~$V_k$ is also~$e^{i\lambda_k}$, so $U_k=e^{i\lambda_k}I=V_k$.
Since $U_k=V_k$ holds for each~$k$, we have arrived at $U=V$.

1$\implies$2:
Let us show that $F^A_\lambda=E^{e^{iA}}_{e^{i\lambda}}$ for all $\lambda\in\R$.
The desired conclusion follows from this claim.

The inclusion $F^A_\lambda\subset E^{e^{iA}}_{e^{i\lambda}}$ is elementary and holds also in the case $E^{e^{iA}}_{e^{i\lambda}}=0$.
But
\begin{equation}
\sum_{0\leq\lambda<2\pi}\dim(F^A_\lambda)
=
\sum_{\lambda\in\R}\dim(E^A_\lambda)
=
N
=
\sum_{0\leq\lambda<2\pi}\dim(E^{e^{iA}}_{e^{i\lambda}}),
\end{equation}
so we must have $\dim(F^A_\lambda)=\dim(E^{e^{iA}}_{e^{i\lambda}})$.
This concludes the proof of the claim.
\end{proof}

\begin{remark}
The matrix exponential map is not a homomorphism (cf.\ the Baker--Campbell--Hausdorff formula), so finding which matrices have the same exponential cannot be simply reduced to finding the kernel of the exponential map.
In the simple case where $[A,B]=0$ it can be checked that $e^{iA}=e^{iB}$ if and only if the spectrum of $A-B$ is contained in~$2\pi\Z$, but neither direction of this result remains true if we assume no commutativity.
\end{remark}

\subsection{Preparations for linearization}

Let us fix an integer $N\geq1$ as before.
Let~$\SH$ and~$\U$ denote the sets of skew-hermitean and unitary matrices of dimension~$N$.
We define the exponential map $\exp\colon\SH\to\U$ by $\exp(A)=e^{A}=\sum_{k=0}^\infty\frac1{k!}A^k$.
For $A\in\SH$, we define the adjoint map $\ad_A\colon\SH\to\SH$ by $\ad_AB=[A,B]$. 
We also define the adjoint map $\Ad_U\colon\SH\to\SH$ by $\Ad_UA=UAU^{-1}$ for $U\in\U$.

Let $\phi(z)=z^{-1}(1-e^{-z})=\sum_{k=0}\frac{(-1)^k}{(k+1)!}z^k$.
Using this power series we can also define the operator $\phi(\ad_A)\colon\SH\to\SH$.

The solution of problem~\ref{prob:matrix-exp} given by proposition~\ref{prop:matrix-exp} is difficult to work with, so we linearize the problem.
For this we will need the derivative of the exponential.
The following two lemmas are well known, but we record them explicitly for the sake of an easy reference.
The rest of this section is devoted to lemmas which we will need in section~\ref{sec:lin-solution}.

\begin{lemma}
\label{lma:dexp}
The derivative of the exponential map is
$
\der\exp_A
=
e^A\phi(\ad_A)
$.
\end{lemma}

\begin{lemma}
\label{lma:ad-Ad}
$\Ad_{\exp(A)}=e^{\ad_A}$ for any $A\in\SH$.
\end{lemma}

\begin{lemma}
\label{lma:ad-decomposition}
For any $A\in\SH$ we have $\SH=\ker\ad_A\oplus\im\ad_A$.
\end{lemma}

\begin{proof}
We can assume~$A$ to be diagonal.
We fix a basis so that we can write~$A$ in the block form
\begin{equation}
\label{eq:A-block}
A
=
i
\begin{pmatrix}
\lambda_1I_{k_1} & & & \\
 & \lambda_2I_{k_2} & & \\
 & & \ddots & \\
 & & & \lambda_mI_{k_m} \\
\end{pmatrix}
\end{equation}
where $\lambda_1<\lambda_2<\dots<\lambda_m$ are the distinct eigenvalues of~$-iA$ and~$I_k$ is the~$k\times k$ identity matrix.

The kernel~$\ker\ad_A$ consists of block diagonal matrices whose blocks are of size $k_1\times k_1,k_2\times k_2,\dots,k_m\times k_m$ (ie.\ with block structure corresponding to that of~$A$).
For a complex~$k_1\times k_2$ matrix~$C$, let us denote
\begin{equation}
\label{eq:off-diagonal}
B_C
=
i
\begin{pmatrix}
0 & C & & \\
C^* & 0 & & \\
 & & \ddots & \\
 & & & 0 \\
\end{pmatrix}.
\end{equation}
It now suffices to show that the matrix~$B_C$, for any~$C$, is in the image of~$\ad_A$; all other matrices with zero block diagonal (corresponding to the block structure of~$A$) can be expressed as sums of matrices of this form.
Now a simple calculation shows that $[A,B_C]=(\lambda_1-\lambda_2)B_{iC}$.
Therefore $B_C=[A,(\lambda_1-\lambda_2)^{-1}B_{-iC}]\in\im\ad_A$.
\end{proof}

\begin{lemma}
\label{lma:ad-exp-decomposition}
For $A\in\SH$ define $F_A\colon\SH\to\SH$ by $F_A=\id-\Ad_{e^{-A}}$.
We always have $\ker\ad_A\cap\im F_A=\{0\}$.
Furthermore, $\SH=\ker\ad_A\oplus\im F_A$ if and only if no two eigenvalues of~$-iA$ differ by a number in $2\pi\Z\setminus\{0\}$.
\end{lemma}

\begin{proof}
Let us write~$A$ again in the block form~\eqref{eq:A-block}.
It is clear that~$\ker\ad_A$ consists of block diagonal matrices (with block structure corresponding to that of~$A$).
Let us see why such matrices are not contained in the image of~$F_A$.

Let~$B_C$ denote the off-diagonal matrix of~\eqref{eq:off-diagonal} as in lemma~\ref{lma:ad-decomposition}.
(Again we can without loss of generality only consider the first two eigenvalues; other pairs can be dealt with similarly.)
Now a calculation shows that
\begin{equation}
\label{eq:vv1}
F_AB_C=B_{(1-e^{i\lambda_2-i\lambda_1})C}
.
\end{equation}
Hence~$\im F_A$ contains only matrices with zero diagonal, and $\ker\ad_A\cap\im F_A=\{0\}$.

Let us then suppose that the condition on eigenvalues of~$-iA$ holds.
Then $1-e^{i\lambda_2-i\lambda_1}\neq0$, so $B_C=F_A(1-e^{i\lambda_2-i\lambda_1})^{-1}B_C\in\im F_A$ for all~$C$.
This implies that $\SH=\ker\ad_A\oplus\im F_A$.

Suppose then that $\SH=\ker\ad_A\oplus\im F_A$.
This implies that for each~$k_1\times k_2$ matrix~$D$ there is a matrix~$C$ of the same dimensions so that $F_AB_C=B_D$.
Due to the identity~\eqref{eq:vv1} this is requires that $e^{i\lambda_2-i\lambda_1}\neq1$, so $\lambda_1-\lambda_2\notin2\pi\Z$.
\end{proof}

\begin{lemma}
\label{lma:dexp-expr}
We have
\begin{equation}
\der\exp_A(B+C)
=
e^AB+[e^A,D]
\end{equation}
for any $A\in\SH$, $B\in\ker\ad_A$ and $C=\ad_AD\in\im\ad_A$.
\end{lemma}

\begin{proof}
Using lemma~\ref{lma:dexp} we have
\begin{equation}
\der\exp_A(B+C)
=
e^A\phi(\ad_A)(B+C).
\end{equation}
Now $\ad_AB=0$ so $\phi(\ad_A)B=B$.
By lemma~\ref{lma:ad-Ad} we have
\begin{equation}
\begin{split}
\phi(\ad_A)C
&=
\phi(\ad_A)\ad_AD
\\&=
(1-e^{-\ad_A})D
\\&=
D-Ad_{\exp(-A)}D
\\&=
D-e^{-A}De^{A}.
\end{split}
\end{equation}
Combining these gives the claimed equation.
\end{proof}

\begin{lemma}
\label{lma:dexp-char}
For any $A\in\SH$ the following are equivalent:
\begin{enumerate}
\item\label{item:1} The derivative~$\der\exp_A$ is a bijection.
\item\label{item:2} All matrices that commute with~$e^A$ commute with~$A$.
\item\label{item:3} For any distinct eigenvalues~$i\lambda$ and~$i\mu$ of~$A$ we have $\mu-\lambda\notin2\pi\Z$.
\end{enumerate}
\end{lemma}

\begin{proof}
\ref{item:1}$\iff$\ref{item:3}:
By lemma~\ref{lma:ad-decomposition} we have $\SH=\ker\ad_A\oplus\im\ad_A$.
Let $B\in\ker\ad_A$ and $C\in\SH$.
By lemma~\ref{lma:dexp-expr} $e^{-A}\der\exp_A(B+\ad_AC)=B+C-e^{-A}Ce^A$.
Thus by lemma~\ref{lma:ad-exp-decomposition}~$\der\exp_A$ is bijective if and only if condition~\ref{item:3} holds.

\ref{item:2}$\iff$\ref{item:3}:
Let us change the basis so that~$A$ comes to the block form
\begin{equation}
A
=
i
\begin{pmatrix}
\lambda_1I_{k_1} & & & \\
 & \lambda_2I_{k_2} & & \\
 & & \ddots & \\
 & & & \lambda_mI_{k_m} \\
\end{pmatrix}
\end{equation}
where $\lambda_1<\lambda_2<\dots<\lambda_m$ are the eigenvalues of~$-iA$ and~$I_k$ is the~$k\times k$ identity matrix.
Then
\begin{equation}
e^A
=
\begin{pmatrix}
e^{i\lambda_1}I_{k_1} & & & \\
 & e^{i\lambda_2}I_{k_2} & & \\
 & & \ddots & \\
 & & & e^{i\lambda_m}I_{k_m} \\
\end{pmatrix}.
\end{equation}
Now it is clear that there are matrices that commute with~$e^A$ but not~$A$ if and only if $e^{i\lambda_i}=e^{i\lambda_j}$ for some $i\neq j$.
But this is equivalent with $\lambda_i-\lambda_j\in2\pi\Z$.
\end{proof}

\begin{lemma}
\label{lma:dexp-singular}
For any $A\in\SH$ the set
\begin{equation}
S_A
=
\{t\in\R;\der\exp_{tA}\text{ is not bijective}\}
\end{equation}
is countable and discrete.
\end{lemma}

\begin{proof}
Let $\sigma\subset\R$ be the set of eigenvalues of~$iA$.
Then by lemma~\ref{lma:dexp-char}
\begin{equation}
S_A
=
\{2\pi z/(\lambda-\mu);z\in\Z\setminus\{0\},\lambda,\mu\in\sigma,\lambda\neq\mu\}
.
\end{equation}
This set is clearly countable and discrete.
\end{proof}

\subsection{Solution of the linearized problem}
\label{sec:lin-solution}

We are now ready to solve a linearized version of our problem.
Physically linearization amounts to assuming the Hamiltonian to be very close to some known Hamiltonian~$H_0$.
We make the additional assumption that~$H_0$ is a constant matrix, independent of the position $x\in\Omega$.
We believe that the linearized problem can be solved in a similar way even for nonconstant~$H_0$ (at least generically if not for every~$H_0$ in every domain).

Denoting by~$H$ the small deviation from~$H_0$, we see from equation~\eqref{eq:vv2} that the linearized problem is to recover~$H$ from the knowledge of
\begin{equation}
\der\exp_{-iTH_0}i\int_0^TH(\gamma(t))\der t
\end{equation}
up to a multiple of~$e^{-iTH_0}$.
We will need a little geometrical lemma to solve this problem.

\begin{lemma}
\label{lma:length}
Let $\Omega\subset\R^n$ be a domain.
Then for every $r\in\R$ the set $A_r=\{(x,y)\in\partial\Omega\times\partial\Omega;d(x,y)=r\}\subset\partial\Omega\times\partial\Omega$ is closed and has no interior points.
\end{lemma}

\begin{proof}
Closedness follows from the continuity of the distance function $d\colon\partial\Omega\times\partial\Omega\to\R$.
Since $A_r=\emptyset$ for $r<0$ and~$A_0$ is the (closed) diagonal of~$\partial\Omega\times\partial\Omega$, it remains to consider $r>0$.

Suppose~$(x,y)$ is an interior point of~$A_r$ for some $r>0$.
Then there is $\delta>0$ so that $d(z,w)=r$ for all $z\in B(x,\delta)\cap\partial\Omega$ and $w\in B(y,\delta)\cap\partial\Omega$.
Fixing $z=x$ shows that near~$y$ the boundary~$\partial\Omega$ is a sphere centered at~$x$; and similarly near~$x$ it is a sphere centered at~$y$.
Let~$L$ be the line containing~$x$ and~$y$.
Let~$L'$ be a line parallel to~$L$ with a distance less than~$\delta$ to it.
Let~$z$ be the unique point in the intersection $L'\cap B(x,\delta)\cap\partial\Omega$ and similarly $w\in L'\cap B(y,\delta)\cap\partial\Omega$.
Then clearly $d(z,w)<r$, a contradiction.
\end{proof}

\begin{theorem}
\label{thm:lin}
Let $\Omega\subset\R^n$, $n\geq2$, be a bounded domain.
Let $H\in C(\bar\Omega,\C^{N\times N})$ take values in hermitean matrices and let~$H_0$ be a fixed hermitean matrix.
Suppose that for every line $\gamma\colon[0,T]\to\bar\Omega$ with endpoints on~$\partial\Omega$ we know
\begin{equation}
e^{iTH_0}\der\exp_{-iTH_0}i\int_0^TH(\gamma(t))\der t
\end{equation}
up to an unknown multiple of the identity matrix.
Then we can recover the trace free part of~$H(x)$ for every $x\in\Omega$.
\end{theorem}

\begin{proof}
Let us denote by~$\mathcal A$ the space of trace free hermitean~$N\times N$ matrices.
We can then decompose the unknown matrix function uniquely as $H(x)=A(x)+h(x)I$ where $A(x)\in\mathcal A$ and $h(x)\in\R$.
Both~$A$ and~$h$ are continuous functions.

By lemmas~\ref{lma:dexp-singular} and~\ref{lma:length} we know that~$\der\exp_{d(x,y)H_0}$ is injective for a dense set of pairs $(x,y)\in\partial\Omega\times\partial\Omega$.
By continuity and density, we can recover the integral
\begin{equation}
\int_0^TA(\gamma(t))\der t
\end{equation}
for every line $\gamma\colon[0,T]\to\bar\Omega$.
(Note that by lemma~\ref{lma:dexp-expr}~$e^{iTH_0}\der\exp_{-iTH_0}$ maps multiples of identity to multiples of identity.)
By injectivity of the X-ray transform in Euclidean domains (see eg.~\cite{book-helgason,book-natterer,radon,cormack}), we can recover~$A$ everywhere in~$\Omega$.
\end{proof}

There may be lines for which the corresponding derivative of the exponential is not injective.
The set of singular lines has no interior points so we were able to use a density argument, but we suspect that this renders recovery somewhat unstable.
A stability result for the linearized problem would give a local uniqueness result for the full problem.

\section{X-ray transforms with matrix weights}
\label{sec:wxrt}

Let~$M$ be the closure of a compact, smooth domain in~$\R^n$.
Let $SM=M\times S^{n-1}$ denote the unit sphere bundle on~$M$ and~$L$ the manifold of maximal directed lines in~$M$, which is the quotient of~$SM$ by the geodesic flow.
Let $\pi_M\colon SM\to M$ and $\pi_L\colon SM\to L$ be the projection and the quotient map.

Let $W\colon SM\to\C^{N\times N}$ be a continuous matrix function.
We define the weighted X-ray transform $\rt_Wf\colon L\to\C^N$ of a continuous function $f\colon M\to\C^N$ by
\begin{equation}
\rt_Wf(\gamma)
=
\int_0^TW(\gamma(t),\dot\gamma(t))f(\gamma(t))\der t
\end{equation}
for any line $\gamma\colon[0,T]\to M$.

We will first restrict our attention to the case $n=2$ and later use results from this case for higher dimensions.

For any $p\in M$ such that~$W(p,v)$ is invertible for all $v\in S^{n-1}$ there is a continuous function $Q_W^p\colon L\to\C^{N\times N}$ so that $(\pi_L^*Q_W^p)W=I$ on the fiber~$S_pM=\{p\}\times S^{n-1}$.
This will allow us to make the weight trivial at any single chosen point $p\in M$ in all directions.

Let~$I_\alpha$ denote the Riesz potential
\begin{equation}
I_\alpha f(x)
=
c_\alpha\int_{\R^2}\frac1{\abs{y}^{2-\alpha}}f(x+y)\der y.
\end{equation}
The size of the constant~$c_\alpha$ is irrelevant for us, but we remark that $c_1=1/2\pi$.
Let~$P$ be the back projection operator defined on functions $g\colon L\to\C^N$ by $Pg\colon M\to\C^N$,
\begin{equation}
Pg(x)
=
\frac1{2\pi}
\int_{S^1}\pi_L^*g(x,v)\der v.
\end{equation}

We will invert~$\rt_W$ for suitable weights~$W$ by freezing the weight at a point near the boundary and using the identity $P\rt_If=I_1f$ and injectivity of the Riesz potential in suitable spaces.
This gives a proof of injectivity in two dimensions in a sufficiently small domain (where the weight is close enough to the identity matrix), which in turn gives a more general injectivity results in dimension three or higher.
(Injectivity in small domains was proved for the attenuated X-ray transform in~\cite{F:x-ray-small-support} with smooth scalar attenuations.)
Our proof is based on the ideas of~\cite{MQ:weighted-xrt}, but we improve it in two important ways.
We allow rougher weights; in particular, we need not assume that the weight has even a first derivative.
Most importantly --- from the point of view of problem~\ref{prob:matrix-wxrt} at least --- we allow the weight to be any invertible matrix.
A similar problem for X-ray transforms with matrix weights has been studied before by several authors~\cite{N:matrix-radon,GN:wxrt-numeric,K:wxrt,N:wxrt-plane,V:matrix-wxrt}.
There is also a recent preprint~\cite{PSUZ:matrix-wxrt} with very general results for matrix weights on manifolds.

Injectivity of~$\rt_W$ requires that the weight~$W$ is sufficiently regular.
Our regularity assumption is quite relaxed, as we do not require differentiability.
Assumptions on the weight are less stringent in dimension three and higher because the problem is formally overdetermined.
There are counterexamples to injectivity if the weight is not regular enough; see~\cite{MQ:weighted-xrt,B:weighted-radon-noninjective}.

\begin{lemma}
\label{lma:wrt-estimate}
Suppose $0\in M\subset\R^2$.
Take any $\beta>0$ and consider a weight~$W$ which satisfies $W\in C^\beta(M,\Lip(S^1,\C^{N\times N}))$ and $W(0,v)=0$ for all $v\in S^1$.
There are $p\in(1,\infty)$ and $\eps>0$ so that for any $\delta>0$ we have for any $f\in L^p(B(0,\delta)\cap M,\C^N)$
\begin{equation}
\aabs{\nabla P\rt_Wf}_p
\leq
C_p\delta^\eps\aabs{W}_{C^\beta(M,\Lip(S^1,\C^{N\times N}))}\aabs{f}_p,
\end{equation}
where~$C_p$ is a constant depending only on~$p$. 
\end{lemma}

\begin{proof}
Let us denote by~$\lesssim$ inequalities including constants that only depend on the exponents.
Let us assume that~$f$ is supported in~$B(0,\delta)$.

Let $i\in\{1,2\}$ and $j=3-i$.
A calculation gives
\begin{equation}
\begin{split}
\partial_iP\rt_Wf(x)
&=
\frac1{2\pi}
\partial_i
\int_{\R^2}\frac1{\abs{y}}W(x+y,y/\abs{y})f(x+y)\der y
\\&=
\frac1{2\pi}
\int_{\R^2}\frac1{\abs{y}}[
\partial_iW(x+y,y/\abs{y})f(x+y)
\\&\quad
+W(x+y,y/\abs{y})\partial_if(x+y)
]\der y
\\&=
\frac1{2\pi}
\int_{\R^2}\bigg[
\frac{y_i}{\abs{y}}W(x+y,y/\abs{y})
\\&\quad
-(-1)^i\frac{y_j}{\abs{y}}\partial_vW(x+y,y/\abs{y})
\bigg]\frac1{\abs{y}}f(x+y)\der y
.
\end{split}
\end{equation}
(We may use a simple approximation argument to justify the above calculation with the assumed regularity of~$W$ and~$f$.)
Therefore, if $K(x)=\frac1{\abs{x}}\sup_{v\in S^1}(\abs{W(x,v)}+\abs{\partial_vW(x,v)})$, we have
\begin{equation}
\abs{\nabla P\rt_Wf(x)}
\lesssim
\int_{B(-x,\delta)}K(x+y)\frac1{\abs{y}}\abs{f(x+y)}\der y.
\end{equation}
By H\"older's inequality,
\begin{equation}
\abs{\nabla P\rt_Wf(x)}
\lesssim
\aabs{K}_{L^{k'}(B(0,\delta))}
(I_{2-k}\abs{f}^k(x))^{1/k}
\end{equation}
for any $k\in[1,2)$.
(The proof given in~\cite{MQ:weighted-xrt} uses this estimate for $k=1$ only, which leads to a more restricting regularity assumption.)
By the Hardy--Littlewood--Sobolev inequality
\begin{equation}
\aabs{I_\alpha g}_s
\lesssim
\aabs{g}_q,
\end{equation}
where $\alpha\in(0,2)$, $q=2s/(2+\alpha s)$ and $s\in[1,\infty)$.
(The requirement $q<2/\alpha$ is always satisfied.)

Combining our estimates with $s=p/k$ and $\alpha=2-k$, we have
\begin{equation}
\begin{split}
\aabs{\nabla P\rt_Wf(x)}_{L^p(\R^2)}
&\lesssim
\aabs{K}_{L^{k'}(B(0,\delta))}
\aabs{I_{2-k}\abs{f}^k}_{L^{p/k}(\R^2)}^{1/k}
\\&\lesssim
\aabs{K}_{L^{k'}(B(0,\delta))}
\aabs{\abs{f}^k}_{L^{q}(\R^2)}^{1/k}
\\&=
\aabs{K}_{L^{k'}(B(0,\delta))}
\aabs{f}_{L^{kq}(B(0,\delta))}
.
\end{split}
\end{equation}
Assuming $kq<p$ (or equivalently $2<(2+p)(2-k)$) H\"older's inequality gives
\begin{equation}
\aabs{f}_{L^{kq}(B(0,\delta))}
\lesssim
\delta^{\frac{2+2p-kp-2k}{kp}}\aabs{f}_{L^p(B(0,\delta))}.
\end{equation}
If $k<2$, the condition $2<(2+p)(2-k)$ is satisfied for sufficiently large~$p$, so that
\begin{equation}
\label{eq:vv3}
\aabs{\nabla P\rt_Wf(x)}_{L^p(\R^2)}
\lesssim
\aabs{K}_{L^{k'}(B(0,\delta))}
\delta^{\frac{2+2p-kp-2k}{kp}}
\aabs{f}_{L^p(B(0,\delta))}.
\end{equation}
We therefore want to have $K\in L^s(M)$ for some $s>2$.

Since $W\in C^\beta(M,\Lip(S^1,\C^{N\times N}))$ and $W(0,v)=0$ for all~$v$, we have $K(x)\lesssim\abs{x}^{\beta-1}$.
If we choose $s\in(2,2/(1-\beta))$, we get $K\in L^s(M)$.
It remains to choose~$p$ so that $2<(2+p)(2-s')$.
Inequality~\eqref{eq:vv3} is now the desired estimate.
\end{proof}

We are now ready to give an answer to problem~\ref{prob:matrix-wxrt}.
Using theorem~\ref{thm:ideal-reduction}, we can determine the trace free part of the Hamiltonian~$H$ from our quantum measurements, provided that $H\in C^{1,\beta}$ for some $\beta>0$.
This was stated as theorem~\ref{thm:main}.

\begin{theorem}
\label{thm:wrt}
Let~$M$ be the closure of a bounded convex domain in~$\R^n$, $n\geq3$, and let $\beta>0$.
Let $W\in C^\beta(M,\Lip(S^{n-1},\C^{N\times N}))$ be invertible at every point on $SM=M\times S^{n-1}$.
(This regularity assumption is true, in particular, if $W\in C^{1,\beta}(SM,\C^{N\times N})$.)
If $f\in C(M,\C^N)$ satisfies $\rt_Wf=0$, then $f=0$.
\end{theorem}

\begin{proof}
Suppose $f\neq0$ but $\rt_Wf=0$.
Take~$p$ in the intersection of~$\spt(f)$ and the boundary of its convex hull so that there is a subtangent plane of the convex hull at~$p$ touching the support only at~$p$.
We shrink~$M$ to be this convex hull.
Let $\nu\in S^{n-1}$ a vector pointing inward from~$p$.

For $t>0$ denote $x_t=p+t\nu$ and suppose~$t$ is so small that $x_t\in\operatorname{int}M$.
Let~$G$ be the set of two dimensional subspaces of the chosen tangent plane at~$p$.
Denote $M_t^T=(x_t+T)\cap M$ for $t>0$ and $T\in G$. 
Let~$\delta_t$ be the maximum of diameters of~$M_t^T$, $T\in G$.
We will consider the weighted X-ray transform of~$f$ on each~$M_t^T$, and for the sake of simplicity we will consider $T\in G$ fixed and suppress it from the notation.
To that end, let us denote the restrictions of~$f$ and~$W$ to~$M_t$ by~$f_t$ and~$W_t$, respectively, and the set of lines on~$M_t$ by~$L_t$.

Define $Q_t\colon S^{1}\to\C^{N\times N}$ by $Q_t(v)=W_t(x_t,v)^{-1}$.
This is a Lipschitz function, and we may define another Lipschitz function $\tilde Q_t\colon L_t\to \C^{N\times N}$ by letting $\tilde Q_t(\ell)=Q_t(v_\ell)$, where $v_\ell\in S^1$ is the direction of the directed line $\ell$.
The Lipschitz constant of~$\tilde Q_t$ is uniformly bounded in~$t$.


Since $\rt_{W_t}f_t=0$, we have $\rt_{(\pi_{L_t}^*\tilde Q_t)W_t}f_t=\tilde Q_t\rt_{W_t}f_t=0$.
Let us denote $\tilde W_t=(\pi_{L_t}^*\tilde Q_t)W_t$.
The function $\pi_{L_t}^*\tilde Q_t\colon M_t\times S^1\to\C^{N\times N}$ is Lipschitz (uniformly in $t$) and only depends on the direction, so we have $\aabs{\tilde W_t}_{C^\beta(M,\Lip(S^1,\C^{N\times N}))}\lesssim1$ uniformly in~$t$ and $\tilde W_t(x_t,v)=I$ for all $v\in S^1$.
Using lemma~\ref{lma:wrt-estimate} we have
\begin{equation}
\begin{split}
C\aabs{f_t}_p
&\leq
\aabs{\nabla I_1f_t}_p
\\&=
\aabs{\nabla P\rt_If_t}_p
\\&=
\aabs{\nabla P\rt_{\tilde W_t}f_t+\nabla P\rt_{I-\tilde W_t}f_t}_p
\\&\leq
\aabs{\nabla P\rt_{\tilde W_t}f_t}_p+\aabs{\nabla P\rt_{I-\tilde W_t}f_t}_p
\\&\leq
\aabs{\nabla P\rt_{\tilde W_t}f_t}_p+c \delta_t^\eps \aabs{f_t}_p
\end{split}
\end{equation}
for a sufficiently large fixed~$p$ and some constants $C,c,\eps>0$.
Note that this estimate is uniform in $T\in G$.
Since $\delta_t\to0$ as $t\to0$, there is $t_0>0$ is such that $c\delta_t^\eps<C/2$ for all $t<t_0$.
For $t<t_0$ we get thus
\begin{equation}
\label{eq:wrt-est}
\frac C2\aabs{f_t}_p\leq \aabs{\nabla P\rt_{\tilde W_t}f_t}_p.
\end{equation}
But $\rt_{\tilde W_t}f_t=0$ so $f_t=0$ for $t<t_0$.
This means that $f=0$ in $\bigcup_{0\leq t<t_0}\bigcup_{T\in G}M_t^T$, which is a neighborhood of~$p$.
Due to our choice of~$p$, this is impossible.
Therefore $\rt_Wf=0$ is only possible if $f=0$.
\end{proof}

It is not actually necessary that~$W(x,v)$ is invertible for all $(x,v)\in SM$.
The proof above only used invertible for~$v$ in a subspace homeomorphic to~$S^{n-2}$ of~$S^{n-1}$, given~$x$.
If one wants to recover a function from its integrals over some set of lines $L'\subset L$, it can only be done stably if the normal bundle of~$L'$ covers~$SM$ (see eg.~\cite{Q:xrt-intro,SU:x-ray-book,FSU:general-x-ray}).

Inspecting the proof of theorem~\ref{thm:wrt}, we see that we have in fact also proven the following local result.
A similar result on manifolds (with $n\geq3$, $N=1$ and $W\in C^\infty$) was shown by Stefanov, Uhlmann and Vasy~\cite[Corollary~3.2]{SUV:partial-bdy-rig}.
As mentioned above, an earlier Euclidean version can be found in~\cite{MQ:weighted-xrt} (for~$C^2$ scalar weights) and in~\cite{SU:x-ray-book} (for~$C^\infty$ scalar weights using microlocal tools).

\begin{theorem}
\label{thm:wrt-local}
Let~$M$ be the closure of a bounded convex~$C^1$ domain in~$\R^n$, $n\geq3$, and let $\beta>0$.
Let $W\in C^\beta(M,\Lip(S^{n-1},\C^{N\times N}))$ be a weight.
Fix $p\in\partial M$.
Suppose~$p$ has a neighborhood~$U$ so that~$W(x,v)$ is invertible for all $x\in U$ and~$v$ in a neighborhood of tangential directions at~$p$.
Suppose for $f\in C(U,\C^N)$ we have $\rt_Wf(\gamma)=0$ for all lines~$\gamma$ that stay in~$U$.
Then there is a possibly smaller neighborhood $U'\subset U$ of~$p$ so that $f|_{U'}=0$.
\end{theorem}

The proof also gives a stability estimate for the local problem; cf.~\eqref{eq:wrt-est}.

\section*{Acknowledgements}

Much of this work was completed during the author's visit to the University of Washington, Seattle, and he is indebted for the hospitality and support offered there; this inculdes financial support from Gunther Uhlmann's NSF grant DMS-1265958.
Part of this work was carried out at Institut Henri Poincar\'e, Paris, with financial support from the institute.
The author was also partially supported by an ERC Starting Grant (grant agreement no 307023).
He is also grateful to Gunther Uhlmann and Mikko Salo for several discussions.
Comments and suggestions from the referees have been most helpful.

\bibliographystyle{siam}
\bibliography{ip}

\end{document}